\documentclass[11pt]{article}

\usepackage{fullpage}
\usepackage{amssymb,amsmath}
\usepackage{graphicx, epsfig}

\def\idtt#1{\ensuremath{\mathtt{#1}}}

\newtheorem{theorem}{Theorem}
\newtheorem{lemma}{Lemma}
\newtheorem{proposition}{Proposition}

\newenvironment{proof}{\trivlist\item[]\emph{Proof}:}%
{\unskip\nobreak\hskip 1em plus 1fil\nobreak$\Box$
\parfillskip=0pt%
\endtrivlist}

{\unskip\nobreak\hskip 1em plus 1fil\nobreak$\Box$
\parfillskip=0pt%
\endtrivlist}

\newenvironment{itemize*}%
  {\begin{itemize}%
    \setlength{\itemsep}{0pt}%
    \setlength{\parskip}{0pt}%
    \setlength{\parsep}{0pt}%
    \setlength{\topsep}{0pt}%
    \setlength{\partopsep}{0pt}%
  }%
  {\end{itemize}}%

\newcommand{\cT}{{\cal T}}

\newcommand{\cD}{{\cal D}}
\newcommand{\cM}{{\cal M}}

\newcommand{\oS}{\overline{S}}
\newcommand{\bT}{\mathbb{T}}

\newcommand{\clab}{\idtt{lab}}
\newcommand{\rank}{\idtt{rank}}
\newcommand{\select}{\idtt{select}}

\newcommand{\Comp}{\idtt{Comp}}
\newcommand{\Tbl}{\idtt{Tbl}}

\newcommand{\eps}{\varepsilon}

\begin{document}

\title{A Dynamic Stabbing-Max Data Structure with Sub-Logarithmic  Query Time}
\author{Yakov Nekrich\thanks{Department of Computer Science, 
University of Chile. % University of Bonn. 
Supported in part by Millennium Institute for Cell Dynamics 
	        and Biotechnology (ICDB), Grant ICM P05-001-F, Mideplan, Chile.
Email {\tt yakov.nekrich@googlemail.com}.}
}
%\institute{}
\date{}
\maketitle
\begin{abstract}
In this paper we describe a dynamic data structure that answers
one-dimensional  stabbing-max queries  in optimal $O(\log n/\log\log n)$ time. 
Our data 
structure uses linear space and supports insertions and deletions in 
$O(\log n)$ and $O(\log n/\log \log n)$ amortized time respectively.  

We also describe a $O(n(\log n/\log\log n)^{d-1})$ space 
data structure that answers $d$-dimensional stabbing-max 
queries in $O( (\log n/\log\log n)^{d})$ time. Insertions and deletions 
are supported in $O((\log n/\log\log n)^d\log\log n)$ 
and $O((\log n/\log\log n)^d)$ amortized time respectively.
\end{abstract}
\section{Introduction}
In the stabbing-max problem, a set of rectangles is stored in a data structure,
and each rectangle $s$ is assigned a priority $p(s)$.
For a query point $q$, we must find the highest priority rectangle $s$ that 
contains (or is stabbed by) $q$. 
%The problem can be also extended to $d>1$ dimensions: the set $S$ contains 
%axis-parallel rectangles instead of intervals.  
In this paper we describe a dynamic data structure that answers 
stabbing-max queries on a set of one-dimensional 
rectangles (intervals) in optimal $O(\log n/\log \log n)$ 
time. We also show how this result can be extended to $d>1$ dimensions.

{\bf Previous Work.} 
The stabbing-max problem has important applications in networking and 
geographic information systems.
Solutions to some special cases of the stabbing-max problem  play a 
crucial role in classification and routing of Internet packets; we 
refer to e.g.,~\cite{GK00,FM00,SK02} for a small selection of the previous 
work and to~\cite{GK00a,SK02a} for more extensive surveys. 
Below we describe the  previous works on the general case of the 
stabbing-max problem.

The semi-dynamic data structure of Yang and Widom~\cite{YW01} maintains a set 
of one-dimensional intervals in linear space; 
 stabbing-max queries  and insertions are supported in $O(\log n)$ time.
Agarwal et al.~\cite{AAYY03} showed that 
stabbing-max queries on a set of one-dimensional intervals can be 
answered in $O(\log^2 n)$ time; the data structure of 
Agarwal et al.~\cite{AAYY03} uses linear space 
and supports updates in $O(\log n)$ time.
The linear space data structure of Kaplan et al.~\cite{KMT03} supports 
one-dimensional queries 
and insertions in  $O(\log n)$ time, but  deletions 
take $O(\log n\log\log n)$ time. In~\cite{KMT03}, the authors also consider 
the stabbing-max problem for a nested set of intervals: for any two 
intervals $s_1,\,s_2\in S$, either $s_1\subset s_2$ or $s_2\subset s_1$ 
or $s_1\cap s_2=\emptyset$. 
Their data structure for a nested set of one-dimensional intervals 
uses $O(n)$ space and 
supports both queries and updates in $O(\log n)$ time. 
Thorup~\cite{Th03} described a linear space data structure that supports
 very fast 
queries, but needs $\log^{\omega(1)}n$ time to perform updates.  
His 
 data structure supports stabbing-max queries in $O(\ell)$ time and updates in 
$O(n^{1/\ell})$ time for any parameter $l=o(\log n/\log \log n)$. 
However the problem of constructing a data structure with poly-logarithmic 
update time is not addressed in~\cite{Th03}. 
%However the situation when updates take poly-logarithmic time is not
% considered 
%in~\cite{Th03}; 
%a direct extension of the result in~\cite{Th03} to 
%the case $\ell=\log n/\log \log n$ would lead 
%to a data structure that supports queries in $O(\log n/\log \log n)$ time 
%and updates in $O(\log^2 n\log\log n)$ time.
Agarwal et al.~\cite{AAY05} described a 
data structure that uses linear space and supports both queries and updates 
in $O(\log n)$ time for an arbitrary set of one-dimensional intervals. 
The results presented in~\cite{KMT03} and~\cite{AAY05} are valid in the 
pointer machine model~\cite{T79}.

The one-dimensional data structures can be extended to $d>1$ dimensions, so 
that space usage, query time, and update time increase by 
 a factor of $O(\log^{d-1} n)$. 
Thus the best previously known data structure~\cite{AAY05} 
for $d$-dimensional stabbing-max 
problem uses $O(n\log^{d-1}n)$ space, answers queries in $O(\log^dn)$ time, 
and supports updates in $O(\log^d n)$ time. Kaplan et al.~\cite{KMT03}
showed that $d$-dimensional stabbing-max queries can be answered in 
$O(\log n)$ time for any constant $d$ in the special case of nested
 rectangles. 

{\bf Our Result.}
The one-dimensional data structure described in~\cite{AAY05} achieves 
optimal query time in the pointer machine model.  In this paper we 
show that we can achieve sublogarithmic query time without increasing the 
update time in the word RAM model of computation.
%a better result can obtained in the word RAM model of computation. 
%We describe a linear space data structure that answers 
%one-dimensional stabbing-max queries in $O(\log n/\log \log n)$ time.  
Our data structure supports deletions and insertions in $O(\log n/\log\log n)$ 
and $O(\log n)$ amortized time respectively. 
As follows from the lower bound of~\cite{AHR98}, any fully-dynamic data
 structure 
with poly-logarithmic update time needs $\Omega(\log n/\log \log n)$ time 
to answer a stabbing-max query\footnote{In fact, the lower bound 
of~\cite{AHR98} 
is valid even for existential stabbing queries: Is there an interval in the set 
$S$ that is stabbed by  a query point $q$?}. 
Thus our data structure achieves optimal query time and space usage.

Our result can be also extended to $d>1$ dimensions. We describe 
a data structure that uses $O(n(\log n/\log\log n)^{d-1})$ space and 
answers stabbing-max queries 
in $O((\log n/\log\log n)^{d})$ time; insertions and deletions 
are supported in $O((\log n/\log\log n)^{d}\log \log n)$ and 
$O((\log n/\log\log n)^{d})$ amortized time respectively.
Moreover, our construction can be modified to support 
stabbing-sum\footnote{The stabbing-sum problem considered in this paper is to
be distinguished from the more general stabbing-group problem, in which every interval is associated with a weight drawn from a group $G$.} 
queries: 
we can count the number of one-dimensional intervals stabbed by a query point 
$q$ in $O(\log n/\log \log n)$ time. 
The stabbing-sum data structure also supports insertions and deletions 
in $O(\log n)$ 
and $O(\log n/\log\log n)$ amortized time.

{\bf Overview.}
We start by describing a simple one-dimensional stabbing-max data structure in 
section~\ref{sec:stabover}. This data structure achieves the desired query 
and update times but needs $\omega(n^2)$ space.  All intervals are stored 
in nodes of the base tree of height $O(\log n/\log \log n)$; the base tree 
is organized as a variant of the 
segment tree data structure. Intervals  in a node $u$ are stored 
in a variant of the van Emde Boas (VEB) data structure~\cite{EKZ77}. 
We can answer a stabbing-max query by traversing a leaf-to-root path in the 
base tree; the procedure spends  $O(1)$ time in each node on the path.

In section~\ref{sec:compact}, we show how all secondary data structures 
in the nodes of the base tree 
can be stored in $O(n\log n)$ bits of space. The main idea of our method 
is the compact representation of intervals stored in each node.  
Similar compact representations were also used in data structures 
for range reporting queries~\cite{Ch88,N09} and  some other 
problems~\cite{B08}. 
However the previous methods are too slow for our goal: 
we need $O(\log\log n)$ time to obtain the 
representation of an element $e$ in a node $u$ if the representation of 
$e$ in a child of $u$ is known. 
Therefore it would take $O(\log n)$ time to traverse a leaf-to-root path in 
the base tree. 
In this paper we present a new, improved compact storage scheme.  Using our 
representation, we can traverse a path in the base tree and spend $O(1)$ time 
in each node.  We believe that our method is of independent interest and 
can be also applied to other problems.

The results for multi-dimensional stabbing-max queries and stabbing-sum 
queries are described 
in Theorems~\ref{theor:multidim} and~\ref{theor:sum1d};
 an extensive description of these results is provided in Appendices~C and~D.
%the multi-dimensional data 
%structure is given in Appendix C. 
%Details of some auxiliary data structures 
%are given in Appendix A. %section~\ref{sec:details}.
%We address the technical issues of re-balancing the base tree in Appendix B.
%The data structure for stabbing-sum queries is described in Appendix D.

\section{A Data Structure with Optimal Query Time}
\label{sec:stabover}
{\bf Base Tree.}
In this section we describe a data structure that answers stabbing-max queries
in optimal time. Endpoints of all intervals from the set $S$ are 
stored in the leaves of the base tree $\cT$. Every leaf of $\cT$ contains 
$\Theta(\log^{\eps} n)$ endpoints. Every internal node, except of the root, 
has $\Theta(\log^{\eps}n)$ children; the root has $O(\log^{\eps}n)$ children. 
Throughout this paper $\eps$ denotes an arbitrarily small 
positive constant.  The range $rng(u)$ of a node $u$ is an interval bounded 
by the minimal and maximal values stored in the leaf descendants of $u$.
 
For a leaf node $u_l$, the set $S(u_l)$ contains all intervals $s$, such that 
at least one endpoint of $s$ belongs to $u_l$. 
For an internal node $u$, the set $S(u)$ contains all intervals $s$, such that 
$rng(u_i)\subset s$ for at least one child $u_i$ of $u$ but 
$rng(u)\not\subset s$. See Fig.~\ref{fig:basetree} for an example.
Thus each interval is stored in $O(\log n/\log\log n)$ sets $S(u)$.
For every pair $i\leq j$, $S_{ij}(u)$ denotes the set of all 
intervals $s\in S(u)$ such that $rng(u_f)\subset s$ for a child $u_f$ of $u$ 
if and only if $i\leq f\leq j$.  
For simplicity, we will sometimes not distinguish between intervals and 
their priorities. 

\begin{figure}[tb]
\centering
\includegraphics[width=.6\textwidth]{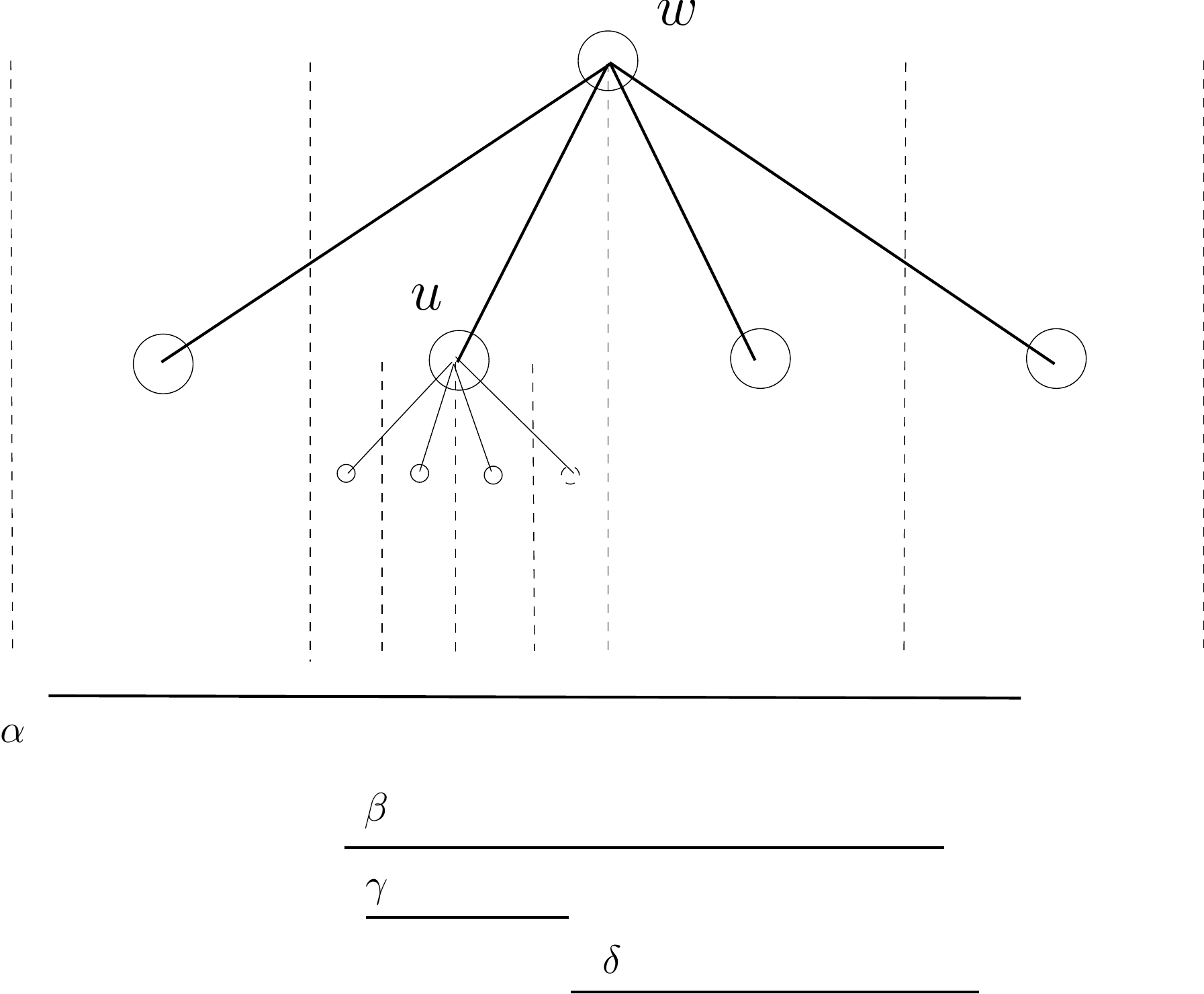}
\caption{\label{fig:basetree} 
Internal nodes of the base tree. Intervals $\alpha$ and $\delta$ are stored
 in the set $S(w)$, but are not stored in the set $S(u)$. Interval 
$\gamma$ is stored in $S(u)$, but $\gamma$ is not stored in $S(w)$. Interval 
$\beta$ is stored in both $S(u)$ and $S(w)$. Moreover, 
$\alpha$, $\beta$, and $\delta$ belong to  the sets 
$S_{23}(w)$, $S_{33}(w)$, and $S_{33}(w)$ respectively.
Intervals $\beta$ and $\gamma$ belong to $S_{24}(u)$ and $S_{23}(u)$ 
respectively.
}
\end{figure}

\tolerance=1000
{\bf Secondary Data Structures.} 
For each internal node $u$, we store a data structure $D(u)$ described in the 
following Proposition.
\begin{proposition}\label{lemma:search1d}
Suppose that priorities of all intervals in $S(u)$ are integers in 
the interval $[1,p_{\max}]$ for $p_{\max}=O(n)$. 
There exists a data structure $D(u)$ that uses $O(p_{\max}\cdot\log^{2\eps}n)$ 
words of space 
and supports the following queries: for any $l\leq r$, 
the predecessor of $q$ in $S_{lr}(u)$ can be found in $O(\log\log n)$
time and the maximum element in $S_{lr}(u)$ can be found in $O(1)$ time. 
The data structure supports insertions in $O(\log\log n)$ time. 
If a pointer to an interval $s\in S(u)$ is given, $s$ can be deleted in 
$O(1)$ amortized time.
\end{proposition}
\begin{proof}
It suffices to store all intervals from $S_{lr}(u)$ in a VEB data structure
$D_{lr}(u)$. Each $D_{lr}(u)$ uses $O(p_{\max})$ words of space and answers 
queries in $O(\log\log p_{\max})=O(\log \log n)$ time~\cite{EKZ77}.  
It is a folklore observation that we can modify the VEB data structure 
so that maximum queries are supported in constant time.
%Divide into blocks. For each block we store $S'_{lr}(v)$ in a 
%VEB data structure. Each element in $S'_{lr}(v)$ contains pointers 
%to the preceding and following elements. Moreover, for each 
%block, we can store pointers to the preceding and the following elements.  
\end{proof}

We also store a data structure $M(u)$ that contains  
the interval $\max_{ij}(u)$ with maximal 
priority among all intervals  in $S_{ij}(u)$ for each $i\leq j$. 
Since each internal node has $\Theta(\log^{\eps}n)$ children, 
$M(u)$ contains $O(\log^{2\eps}n)$ elements. 
For any query index $f$, $M(u)$ reports the largest  element among all 
$\max_{ij}$, $i\leq f\leq j$. In other words for any child $u_f$ of $u$, 
$M(u)$ can find the interval with the highest priority that covers the 
range of the $f$-th child of $u$. 
Using standard techniques, we can implement $M(u)$ so that 
queries and updates are supported in $O(1)$ time. For completeness, 
we describe the data structure $M(u)$ in Appendix A. 
%section~\ref{sec:details}.
% \begin{figure}[tb]
% \centering
% \includegraphics[width=.6\textwidth]{basetree}
% \caption{\label{fig:basetree} 
% Internal nodes of the base tree. Intervals $\alpha$ and $\delta$ are stored
%  in the set $S(w)$, but are not stored in the set $S(u)$. Interval 
% $\gamma$ is stored in $S(u)$, but $\gamma$ is not stored in $S(w)$. Interval 
% $\beta$ is stored in both $S(u)$ and $S(w)$. Moreover, 
% $\alpha$, $\beta$, and $\delta$ belong to  the sets 
% $S_{23}(w)$, $S_{33}(w)$, and $S_{33}(w)$ respectively.
% Intervals $\beta$ and $\gamma$ belong to $S_{24}(u)$ and $S_{23}(u)$ 
% respectively.
% }
% \end{figure}

{\bf Queries and Updates.}
Let $\pi$ denote the search path for the query point $q$ in the base tree 
$\cT$. The path $\pi$ consists of the nodes $v_0,v_1,\ldots,v_R$ where 
$v_0$ is a leaf node and $v_R$ is the root node. 
Let $s(v_i)$ be the interval with the highest priority among all intervals 
that are stored in $\cup_{j\leq i} S(v_j)$ and are stabbed by $q$.
The interval $s(v_0)$ can be found by examining  all 
$O(\log^{\eps}n)$ intervals stored in $S(v_0)$. 
Suppose that we reached a node $v_i$ and the interval $s(v_{i-1})$ 
is already known. 
If $q$ stabs an interval $s$ stored in $S(v_i)$, then $rng(v_{i-1})\subset s$.
Therefore $q$ stabs an interval $s\in S(v_i)$ if and only if  $s$ is stored 
in some set $S_{lr}(v_i)$ such that $l\leq f\leq r$ and $v_{i-1}$ is 
the $f$-th child of $v_i$. 
Using the data structure $M(v_i)$, 
we can find in constant time the maximal interval $s_m$, such that 
$s_m\in S_{lr}(v_i)$ and  $l\leq f\leq r$. Then we just 
set $s(v_i)=\max(s_m,s(v_{i-1}))$ and proceed in the next node $v_{i+1}$. 
The total query time is $O(\log n/\log \log n)$.

When we insert an interval $s$, we identify $O(\log n/\log \log n)$
nodes $v_i$ such that $s$ is to be inserted into $S(v_i)$.  
For every such $v_i$, we proceed as follows. We identify $l$ and $r$ such 
that $s$ belongs to $S_{lr}(v_i)$. Using $D(v_i)$, we find the position of 
$s$ in $S_{lr}(v_i)$, and insert $s$ into $S_{lr}(v_i)$. 
If  $s$ is the maximal interval in 
$S_{lr}(v_i)$, we delete the old interval 
$\max_{lr}(v_i)$ from  the data structure $M(v_i)$, 
set $\max_{lr}(v_i)=s$,  and insert the new $\max_{lr}(v_i)$ into $M(v_i)$.  

When an interval $s$ is deleted, we also identify nodes $v_i$, such that 
$s\in S(v_i)$. For each $v_i$, we  find the indices $l$ and $r$, 
such that $s\in S_{lr}(v_i)$. Using the procedure that will be described 
in the next section, we can find the position of $s$ in $S_{lr}(v_i)$. 
Then, $s$ is deleted from the data structure $D(v_i)$. If $s=\max_{lr}(v_i)$, 
we remove  $\max_{lr}(v_i)$ from $M(v_i)$, find the maximum priority 
interval in $S_{lr}(v_i)$, and insert it into $M(v_i)$. 
We will show in section~\ref{sec:compact} that positions of the 
deleted interval $s$ in $S_{lr}(v_i)$ for all nodes $v_i$ can be found 
in $O(\log n/\log\log n)$ time. Since all other operations take 
$O(1)$ time per node, the total time necessary for a deletion of an 
interval is $O(\log n/\log \log n)$.

Unfortunately, the space usage of the data structure described in this 
section is very high: every VEB data structure $D_{lr}(u)$ needs 
$O(p_{\max})$ space, where $p_{\max}$  is the highest possible interval 
priority. Even if $p_{\max}=O(n)$, all data structures $D(u)$ 
use $O(n^2\log^{2\eps}n)$ space. 
In the next section we show how all data structures $D(u)$, $u\in \cT$, can be 
stored in $O(n\log n)$ bits without increasing the query and update 
times.

\section{Compact Representation}
\label{sec:compact}
The key idea of our compact representation is to store only interval
identifiers in every node $u$ of $\cT$. Our storage scheme enables us to
spend $O(\log\log n)$ bits for each identifier stored in a node $u$.
Using the position of an interval $s$ in a node $u$, we
can obtain the position of $s$ in the parent $w$ of $u$.  We can also
compare priorities of two intervals stored in the same node by
comparing their positions.  These properties of our storage scheme
enable us to traverse the search path for a point $q$ and answer the
query as described in section~\ref{sec:stabover}.

Similar representations were also used in space-efficient data structures 
for orthogonal range reporting~\cite{Ch88,N09} and orthogonal point 
location and line-segment intersection problems~\cite{B08}. 
Storage schemes of~\cite{N09,B08} also use $O(\log\log n)$ bits for each 
interval stored in  a node of the base tree.   
The main drawback of those methods is that we need $O(\log\log n)$ time 
to navigate between a node and its parent. Therefore, $\Theta(\log n)$ 
time is necessary to traverse a leaf-to-root path and we need 
$\Omega(\log n)$ time to answer a query. 
In this section we describe a new method that enables us to navigate between
 nodes of the base tree and update the lists of identifiers 
in $O(1)$ time per node. The main idea of our improvement is to maintain 
identifiers for a set $\oS(u)\supset S(u)$ in every $u\in \cT$. 
When an interval is inserted in $S(u)$, we also add its identifier to
$\oS(u)$. But when an interval is deleted from $S(u)$, its identifier 
is not removed from  $\oS(u)$. When the number of deleted interval identifiers 
in all $\oS(u)$ exceeds the number of intervals in $S(u)$, we re-build 
the base tree and all secondary structures (global re-build). 

%We maintain a set $\oS(u)\supset S(u)$ in every $u\in \cT$. 
%$\oS(u)$ is the set of all intervals whose identifiers are stored in  
%the node $u$. 

{\bf Compact Lists.}
We start by defining a set $S'(u)\supset S(u)$. If $u$ is a leaf 
node, then $S'(u) =S(u)$.
If $u$ is an internal node, then $S'(u)=S(u)\cup (\cup_i S'(u_i))$
for all children $u_i$ of $u$. An interval $s$ belongs to $S'(u)$ 
if at least one endpoint of $s$ is stored  in a leaf descendant of $u$. 
Hence, $|\cup_v S'(v)|=O(n)$ where the union is taken over all nodes $v$ 
that are situated on the same level of the base tree $\cT$.
Since the height of $\cT$ is $O(\log n/\log\log n)$, the total number 
of intervals stored in all $S'(u)$, $u\in \cT$, is $O(n\log n/\log\log n)$.

Let $\oS(u)$ be the set that contains all intervals from $S'(u)$ that were 
inserted into $S'(u)$ since the last global re-build. 
%If $u$ is a leaf 
%node, then $Lab(u) =\oS(u)$.
%If $u$ is an internal node, then $Lab(u)=(\cup_i Lab(u_i))\cup \oS(u)$
%for all children $u_i$ of $u$.
We will organize global re-builds in such way that at most one half of 
elements in all $\oS(u)$ correspond to deleted intervals. 
Therefore the total number of intervals in $\cup_{u\in T} \oS(u)$ is 
$O(n\log n/\log \log n)$. We will show below how we can store the represntations of sets $\oS(u)$ in 
compact form, so that an element of $\oS(u)$ uses $O(\log \log n)$ bits 
in average. Since $S(u)\subset \oS(u)$, we can also use the same method 
to store all $S(u)$ and $D(u)$ in $O(n\log n)$ bits.

Sets $S'(u)$ and $\oS(u)$ are not stored explicitly in the data structure. 
Instead, we store a list $\Comp(u)$ that contains a compact representation 
for identifiers  
of intervals in $\oS(u)$. $\Comp(u)$ is organized as follows.
The set $\oS(u)$ is sorted by interval priorities and divided into blocks. 
If $|\oS(u)|>\log^3n/2$, then  each block of $\oS(u)$ contains 
at least $\log^3n/2$ and at most $2\log^3 n$ elements. Otherwise all 
$e\in \oS(u)$ belong to the same block. Each block $B$ is assigned an 
integer block label $\clab(B)$
according to the method of~\cite{IKR81,W92}. 
Labels of blocks are monotone with respect to order of blocks in 
$\Comp(u)$: The block $B_1$ precedes $B_2$ in $\Comp(u)$ if and only if 
$\clab(B_1)<\clab(B_2)$.
Besides that, all labels assigned to blocks of $\Comp(u)$  are bounded by a 
linear function of $|\Comp(u)|/\log^3n$: for any 
block $B$ in $\Comp(u)$, $\clab(B)=O(|\Comp(u)|/\log^3n)$.  
When a new block is inserted into a list, we may have to change the labels 
of $O(\log^2 n)$ other blocks.

 For every block $B$, we store its block label as well as the pointers 
to the next  and the previous blocks in the list $\Comp(u)$.  
For each interval $\bar{s}$ in a block $B$ of $\oS(u)$, 
the list $\Comp(u)$ contains the \emph{identifier} of $\bar{s}$ in 
$\Comp(u)$. The identifier is simply the list of indices of 
children $u_i$ of $u$, 
such that $\bar{s}\in \Comp(u_i)$. As follows from the description of 
the base tree and the sets $\oS(u)$, we store at most two child indices 
for every interval $\bar{s}$ in a block; hence, each identifier uses 
$O(\log\log n)$ bits. 
To simplify the description, we sometimes will 
not distinguish between an interval $s$ and its identifier in a 
list $\Comp(u)$.

We say that the position of an interval $s$ in a list $\Comp(u)$ 
is known if the block $B$ that contains $s$ and the position of $s$ 
in $B$ are known. If we know positions of two intervals $s_1$ and $s_2$ in $\Comp(u)$, 
we can compare their priorities in $O(1)$ time. 
Suppose that $s_1$ and $s_2$ belong to blocks $B_1$ and $B_2$ respectively. 
Then $s_1>s_2$ if $\clab(B_1)>\clab(B_2)$, and 
$s_1<s_2$ if $\clab(B_1)<\clab(B_2)$. If $\clab(B_1)=\clab(B_2)$, 
we can compare priorities of $s_1$ and $s_2$ by comparing their 
positions in the block $B_1=B_2$.

The rest of this section has the following structure. 
First, we describe auxiliary data structures that enable us to search 
in a block and navigate between nodes of the base tree.
Each block $B$ contains a poly-logarithmic number of elements and every 
identifier in $B$ uses $O(\log \log n)$ bits. We can use this fact 
and implement block data structures, so that queries and updates 
are supported in $O(1)$ time.  
If the position of some $\bar{s}$ in a list $\Comp(u)$ is known, 
we can find in constant time the positions of $\bar{s}$ in the parent of 
$u$.
Next, we show how data structures $D(u)$ and $M(u)$, defined in
 section~\ref{sec:stabover}, are modified. 
Finally, we describe the search and update procedures.

% (1) If an identifier of $\bar{s}$ in $\Comp(u)$ and 
% its position in $\Comp(u)$ are known, we can traverse the path from 
% $u$ to the root of the base tree and obtain the interval $\bar{s}$. 
% We will show how this can be done later in this section. \\

{\bf Block Data Structures.}
We store a data structure $F(B)$ that supports rank and select queries 
in a block $B$ of $\Comp(u)$: A query $\rank(f,i)$ returns the number of 
intervals that also belong to $\oS(u_f)$ among the first $i$ elements of $B$.
A query $\select(f,i)$ returns the smallest $j$, such 
that $\rank(f,j)=i$; in other words, $\select(f,i)$ returns the position 
of the $i$-th interval in $B$ that also belongs to $\oS(u_f)$.
We can answer $\rank$ and $\select$ queries in a block in $O(1)$ time. 
We can also count the number of elements in a block of $\Comp(u)$ 
that are stored in the $f$-th child of $u$, and determine for the 
$r$-th element of a block in which children of $u$ it is stored. 
Implementation of $F(B)$ will be described in Appendix A. 
%section~\ref{sec:details}.

For each block $B_j\in \Comp(u)$ and for any child $u_f$ of $u$, we store 
a pointer to the largest block $B^f_j$ before $B_j$  that contains an element
 from $\Comp(u_f)$. These pointers 
are 
maintained with  help of a data structure $P_f(u)$ for each child $u_f$. 
We implement $P_f(u)$ as an incremental split-find data structure~\cite{IA87}.
Insertion of a new block label into $P_f(u)$ takes $O(1)$ amortized time;
we can also find the block $B^f_j$ for any block $B_j\in \Comp(u)$ in 
$O(1)$ worst-case time.  
Using the data structure $F(B_j)$ for a block $B_j$, 
we can identify for any element $e$ in a block 
$B_j$ the largest $e_f\leq e$, $e_f\in B_j$, such that $e_f$ belongs 
to $\Comp(u_f)$. Thus we can identify for any $e\in \Comp(u)$ the largest
 element $e_f\in \Comp(u)$, such that $e_f\leq e$ and  $e_f\in \Comp(u_f)$,  
in $O(1)$ time. 

{\bf Navigation in Nodes of the Base Tree.}
Finally, we store  pointers to selected elements in each list $\Comp(u)$.  
Pointers enable us to navigate between nodes of the base tree: if 
the position of some $e\in \Comp(u)$ is known, we can find the position 
of $e$  in $\Comp(w)$ for the parent $w$ 
of $u$ and the position of $e$ in $\Comp(u_i)$ for any child $u_i$ of $u$ 
such that $\Comp(u_i)$ contains $e$. 
%Finally, we store additional data for selected elements  $e\in \Comp(u)$ that 
%enables us to find the position of $e$ in $\Comp(w)$ for the parent $w$ 
%of $u$ and the position of $e$ in $\Comp(u_i)$ for any child $u_i$ of $u$ 
%such that $\Comp(u_i)$ contains $e$. 

We associate a {\em stamp} $t(e)$ with each element stored in a block $B$; 
every  stamp is a positive integer bounded by $O(|B|)$. 
A pointer to an element $e$ in a block $B$ consists of the block label 
$\clab(B)$ and the stamp of $e$ in $B$. 
When an element is inserted into a block, stamps of other elements in this 
 block do not change.  Therefore, when a new interval is inserted into a block 
$B$ we do not have to update all pointers that point into $B$. 
Furthermore, we store a data structure $H(B)$ for each block $B$. 
Using $H(B)$, we can find the position of an element $e$ in $B$ if its stamp 
$t(e)$ in $B$ is known. If an element $e$ must be inserted into $B$ after 
an element $e_p$, then we can assign a stamp $t_e$ to $e$ and insert it into 
$H(B)$ in $O(1)$ time. Implementation of $H(B)$ is very similar to the 
implementation of $F(B)$ and will be described in the full version.

If $e$ is the first element in the block $B$ of $\Comp(u)$ 
that belongs to $\Comp(u_i)$, 
then we store the pointer from the copy of $e$ in $\Comp(u)$ 
to the copy of $e$ in $\Comp(u_i)$. 
If an element $e\in \Comp(u)$ is also stored in $\Comp(u_i)$ 
and $e$ is the first element in a block $B'$ of $\Comp(u_i)$,
then there is a pointer from the copy of $e$ in $\Comp(u)$ to the copy 
of $e$ in $\Comp(u_i)$. 
Such pointers will be called child pointers.
For any pointer from  $e\in \Comp(u)$ to $e\in \Comp(u_i)$, we store 
a pointer from $e\in \Comp(u_i)$ to $e\in \Comp(u)$.
Such pointers will be called parent pointers.  See Fig.~\ref{fig:pointers}
for an example. 
\begin{figure}[tb]
\centering
\includegraphics[trim= 0 65mm 0mm 0, clip, width=.7\textwidth]{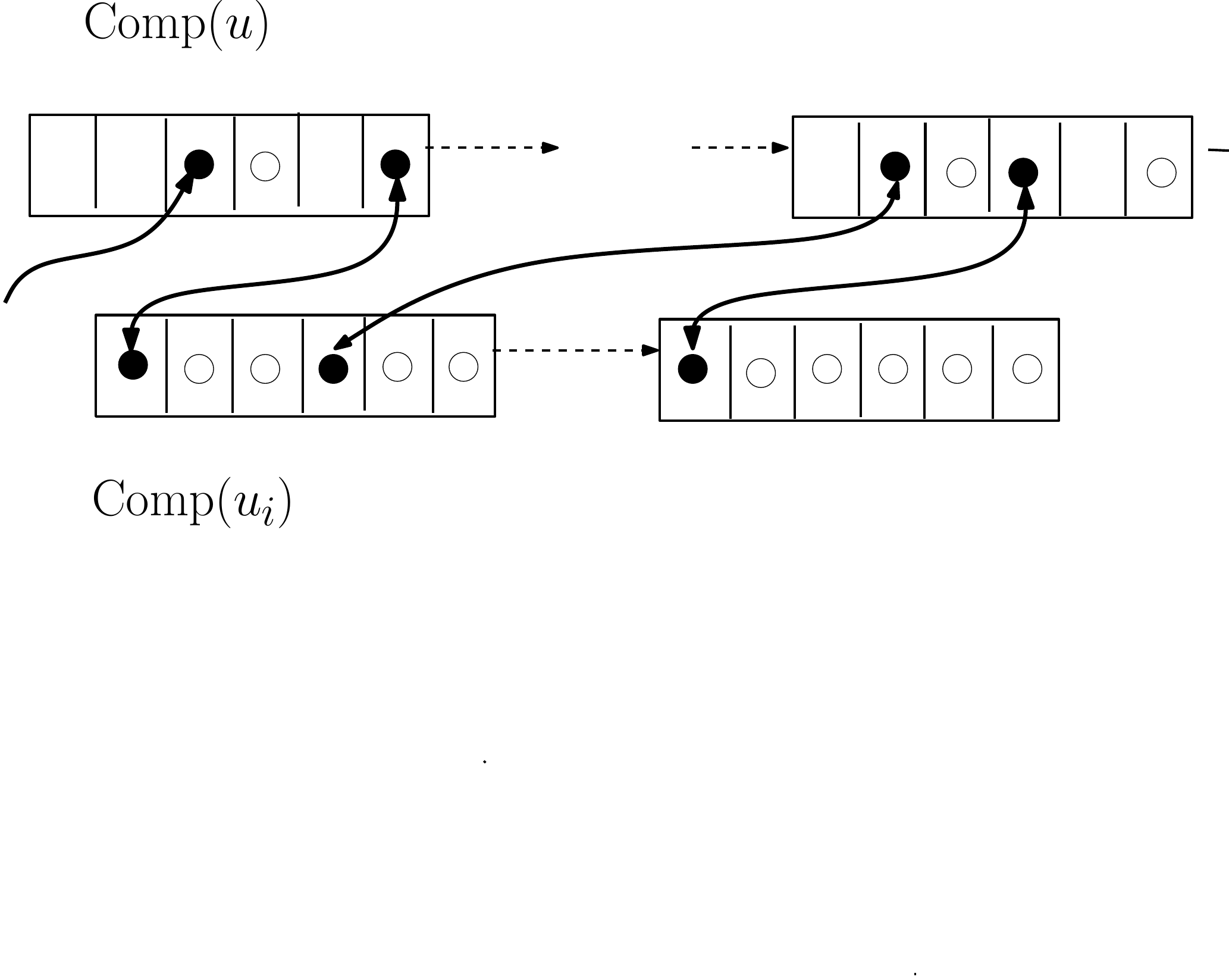}
\caption{\label{fig:pointers} 
Parent and child pointers in blocks of $\Comp(u)$ and $\Comp(u_i)$. 
Intervals that are stored in $\Comp(u_i)$ are depicted by circles; filled 
circles are intervals with pointers. Only relevant intervals and pointers are 
shown in $\Comp(u)$. In $\Comp(u_i)$, only parent pointers are shown.
}
\end{figure}

%
%If $e'$ is the first element in a block $B'$ of $Comp(u_i)$, 
%then there is a pointer from the copy of $e'$ in $\Comp(u_i)$ to the 
%copy of $e'$ in $\Comp(u)$. 
%Such pointers will be called parent pointers. 
%Conversely, if  there is a pointer to a block $B$  
%from a block $B'$ in a parent $w$ of $u$, then we store a pointer 
%to $B'$ in the block $B$. 
%Such pointers will be called child pointers and parent pointers respectively.
We can store each pointer in $O(\log n)$ bits. 
The total number of  pointers 
in $\Comp(u)$ equals to the number of blocks in $\Comp(u)$ and $\Comp(u_i)$ 
for all children $u_i$ of $u$. Hence, all pointers and all block data 
structures use $O(n\log n)$ bits.

%\begin{figure}[tb]
%\centering
%\includegraphics[width=.5\textwidth]{pointers}
%\caption{\label{fig:pointers} 
%Test figure
%}
%\end{figure}

If we know the position of some interval $s$ in $\Comp(v)$, we can 
find the position of $s$ in the parent $w$ of $v$ as follows. 
Suppose that an interval $s$ is stored in the block $B$ of $\Comp(v)$
and $v$ is the $f$-th child of $w$. 
We find the last interval $s'$ in $B$ stored before $s$, such that 
there is a parent pointer from $s'$. 
Let $m$ denote the number of elements between $s'$ and $s$ in $B$.
Let $B'$ be the block in $\Comp(w)$ that contains $s'$. 
Using $H(B')$, we find the position $m'$ of $s'$ in the block $B'$ 
of $\Comp(w)$. Then the position of $s$ in $B'$ can be found by answering 
the query $\select(f,\rank(f,m')+m)$ to a data structure $F(B')$. 
%Then $s$ is the $(m+m')$-th element in $B'$.
Using a symmetric procedure, we can find the position of $s$ in $\Comp(v)$ 
if its position in $\Comp(w)$ is known.

% Let $B'$ be the  block with the largest label, such  that
%  $\clab(B')\leq \clab(B)$  and $B'$ 
% contains the parent pointer. Each block contain less than $2\log^3 n$ 
% and more than $\log^3 n/2$ elements. Therefore either $B'=B$ or 
%  $B'$ is one of the four blocks 
% that immediately precede $B$. Let $n_s$ be the total number of elements between 
% the first element in $B'$ and $s$ in $\Comp(v)$. We can find 
% $n_s$ in $O(1)$ time. Suppose that the parent pointer points to a block 
% $B_p$ in $\Comp(w)$ and $v$ is the $f$-th child of $w$. 
% Then $s$ belongs to the block $B_p$ of $\Comp(w)$, and 
% the position of $s$ in $B_p$ can be found with a query
% $\select(f,n_s)$. 

%The (global) pointer to an element $e$ in  $\Comp(u)$, $ptr(e,\Comp(u))$, 
% consists of the pointer to a block that contains $e$ and the local 
%label of $e$. 

{\bf Root and Leaves of the Base Tree.}
Elements of $\oS(v_R)$ are explicitly stored in the root node $v_R$.
That is, we store a table in the root node $v_R$ that enables us to find 
for any interval $s$ the block $B$ that contains the identifier of $s$
in $\Comp(v_R)$ and the stamp of $s$ in $B$. 
Conversely, if the position of $s$ in a block $B$ of $\Comp(v_R)$ is 
known, we can find its stamp in $B$ in $O(1)$ time. 
If the block label and the stamp of an interval $s$ are known, we
can identify $s$ in $O(1)$ time. 
In the same way, we explicitly store the elements of $\oS(u_l)$ 
for each leaf node $u_l$. 
Moreover, we store all elements of $\oS(v_R)$ in a data structure $R$,
so that the predecessor and the successor of any value $x$ can be found 
in $O(\log n/\log\log n)$ time.

% All intervals that belong to the same block are also stored in a 
% binary tree; hence, we can find the number of elements $s'< s$ in $B$ 
% and identify the position of $s$ in $B$ in $O(\log \log n)$ time. 
%  Conversely, if the position of $s$ in $\Comp(v_R)$ is known, 
% we can find the interval $s$ in $O(\log\log n)$ time. 
% Moreover, we store all elements of $\oS(v_R)$ in a data structure,
% so that the predecessor and the successor of any value $x$ can be found 
% in $O(\log n/\log\log n)$ time.  
% If $u$ is not the root node, we only store the list $\Comp(u)$ in the node 
% $u$.  

%{\bf Data Structures $M(u)$ and $D(u)$.}
%We will show in Appendix A how all data structures $D(u)$ and $M(u)$ 
%can be stored in $O(n\log n)$ bits of space.
%Both $D(u)$ and $M(u)$ contain information about blocks 
{\bf Data Structures $M(u)$ and $D(u)$.}
Each set $S(u)$ and data structure $D(u)$ are also stored in compact form. 
$D(u)$ consists of structures $D_{lr}(u)$ for every pair $l\leq r$. 
If a block $B$ contains a label of an interval $s\in S_{lr}(u)$, then 
we store the label of  $B$ in the VEB data structure $D_{lr}(u)$.  
The data structure $G(B)$ contains data about intervals in  
$S(u)\cap B$.  
For every $s\in G(B)$, we store indices $l,r$ if $s\in S_{lr}(u)$; 
if an element $s\in B$ does not belong to $S$ (i.e., $s$ was already 
deleted), then we set $l=r=0$. 
For a query index $f$, $G(B)$ returns  in $O(1)$ time the highest priority
 interval $s\in B$,
such that $s\in S_{lr}(u)$ and $l\leq f\leq r$.
A data structure $G(B)$ uses $O(|B|\log\log n)$ bits of space and supports
updates in $O(1)$ time.
Implementation of $G(B)$ is very similar to the implementation of $F(B)$ and 
will be described in the full version of this paper. 
We store a data structure $M(u)$ in  every node $u\in \cT$, such that 
$\Comp(u)$ consists of at least two blocks.
For each pair $l\leq r$, the data structure $M(u)$ stores 
 the label of the block that contains the highest priority interval in 
$S_{lr}(u)$. For a query $f$, $M(u)$ returns the label of the block that 
contains  the highest priority interval in $\cup_{l\leq f\leq r} S_{lr}(u)$.
Such queries are supported in $O(1)$ time; a more detailed description 
of the data structure $M(u)$ will be given in Appendix A. 
%section~\ref{sec:details}. 

{\bf Search Procedure.}  
We can easily modify the search procedure of section~\ref{sec:stabover}
for the case when only lists $\Comp(u)$  are stored in each node. 
Let $v_i$ be a node on the 
path $\pi=v_0,\ldots v_R$, where $\pi$ is the search path for 
the query point $q$.  
Let $s(v_i)$ denote  the interval that has the highest priority
 among all intervals that belong to $\cup_{j\leq i}S(v_j)$ and are stabbed by
 $q$.    
We can examine all intervals in $S(v_0)$ and find $s(v_0)$ in 
$O(\log^{\eps}n)$ time. The position of $s(v_0)$ in $\Comp(v_0)$ 
can also be found in $O(1)$ time.  
During the $i$-th step, $i\geq 1$, we find the position 
of $s(v_i)$ in $\Comp(v_i)$.  An interval $s$ in $S(v_i)$ is stabbed by $q$ 
if and only if $s\in S_{lr}(v_i)$, $l\leq f\leq r$, and $v_{i-1}$ is the 
$f$-th child of $v_i$. Using $M(v_i)$, we can find the block label of 
the maximal interval $s_m$ among all intervals stored in $S_{lr}(v_i)$ 
for $l\leq f\leq r$. Let $B_m$ be the block that contains $s_m$. 
Using the data structure $G(B_m)$, we can find the position of $s_m$ 
in $B_m$. 

By definition, $s(v_i)=\max(s_m,s(v_{i-1}))$. Although we have no access to $s_m$ 
and $s(v_{i-1})$, we can 
compare their priorities by comparing positions of $s_m$ and $s(v_{i-1})$ 
in $\Comp(v_i)$. Since the position of $s(v_{i-1})$ in $\Comp(v_{i-1})$ 
is already known, 
we can find its position in $\Comp(v_i)$ in $O(1)$ time. We can also determine, 
whether $s_m$ precedes or follows $s(v_{i-1})$ in $\Comp(v_i)$ in $O(1)$ time.
Since a query to $M(v_i)$ also takes $O(1)$ time, our search procedure 
spends constant time in each node $v_i$ for $i\geq 1$. 
When we know the position of $s(v_R)$ in $\Comp(v_R)$, we can find 
the interval $s(v_R)$ in $O(1)$ time.  
The interval $s(v_R)$ is the highest priority interval in $S$ that is stabbed 
by $q$. 
Hence, a query can be answered in $O(\log n/\log \log n)$ time.

We describe how our data structure can be updated in
 section~\ref{sec:rebal}. 
Our result is summed up in the following Theorem.
\begin{theorem}
\label{theor:1d}
There exists a linear space data structure that answers orthogonal 
stabbing-max queries in $O(\log n/\log \log n)$ time. This data structure 
supports insertions and deletions in $O(\log n)$ and $O(\log n/\log \log n)$ 
amortized time respectively.
\end{theorem}

\section{Updates in the Data Structure of Theorem~\ref{theor:1d}}
\label{sec:rebal}

Suppose that an interval $s$  is inserted into $S$. 
The insertion procedure consists of two parts. 
First, we insert $s$ into lists $\Comp(v_i)$  
and update  $\Comp(v_i)$ for all relevant nodes $v_i$ of $\cT$.
Then, we insert $s$ into  data structures $D(v_i)$ and $M(v_i)$. 

Using the data structure $R$, we can identify the segment $s'(v_R)$
that precedes $s$ in $\oS(v_R)$. Then, we find the position of $s'(v_R)$ in
$\Comp(v_R)$.  We also find the leaves $v_a$ and $v_b$ of $\cT$, in
which the left and the right endpoints of $s$ must be stored.

The path $\pi=v_R,\ldots, v_a$ is
traversed starting at the root node.  Let $s'(v_i)$ be the segment that 
precedes $s$ in $\Comp(v_i)$ and let $B(v_i)$ be the block that contains 
$s'(v_i)$. In every node $v_i$, we insert
the identifier for $s$ after $s'(v_i)$. We also update data 
structures $F(B(v_i))$ and $H(B(v_i))$ for the block $B(v_i)$ that contains 
$s$.  If the number of intervals in
$B(v_i)$ equals $2\log^3n$, we split the block $B(v_i)$ into two blocks
$B_1(v_i)$ and $B_2(v_i)$ of equal size.  We assign a new label to
$B_2(v_i)$ and update the labels for some blocks of $\Comp(v_i)$ 
in $O(\log^2n)$ time.
We also may have to update $O(\log^{\eps}n)$ split-find data
structures $P_f(v_i)$.  Besides that, $O(1)$ child pointers in the parent of
$v_i$ and $O(\log^{\eps}n)$ parent pointers in the children of $v_i$
may also be updated.  We split the block after $\Theta(\log^3 n)$
insertions; hence, an amortized cost of splitting a block is $O(1)$.
When $s$ is inserted into $\Comp(v_i)$, we find the largest element
$s'(v_{i-1})\leq s$, such that $s'(v_{i-1})$ is also stored in
$\Comp(v_{i-1})$. 
Then, we find the position of $s'(v_{i-1})$ in
$\Comp(v_{i-1})$ and proceed in the node $v_{i-1}$.

 Let $v_c$ be the
lowest common ancestor of $v_a$ and $v_b$, and let $\pi_b$ be the path
from $v_c$ to $v_b$.  
We also insert $s$ into $\Comp(v_j)$ for all $v_j\in\pi_b$; 
nodes $v_j\in\pi_b$ are processed in the same way as nodes $v_i\in \pi_a$

During the second stage, we update data structures $D(v_i)$ and $M(v_i)$ 
for $v_i\in \pi_a\cup\pi_b$. For any such $v_i$ we already know the 
block $B\in \Comp(v_i)$ that contains $s$ and the position of 
$s$ in $B$. Hence, a data structure $G(B)$ can be updated in $O(1)$ 
time.
If $s$ is the only interval in $B$  that belongs to 
$S_{lr}(v_i)$ for some $l$, $r$, we insert $\clab(B)$ into 
$D_{lr}(v_i)$. 
If $\clab(B)$ is the greatest label in $D_{lr}(v_i)$,
we also update $M(v_i)$. 
An insertion into a data structure $D_{lr}(v_i)$, $v_i\in \pi_a\cup \pi_b$, 
takes $O(\log\log n)$ time. Updates of all other structures in 
$v_i$ take $O(1)$ time. Hence, an insertion takes $O(\log n)$ time in total.

When an interval is deleted, we identify its position in every relevant node 
$v_i$. This can be done in the same way as during the first stage of
 the insertion procedure.  Then, we update the data structures
$D(v_i)$, $M(v_i)$, and $G(B(v_i))$ if necessary. 
Since a block label can be removed 
from $D_{lr}(v_i)$ in $O(1)$ time, the total time necessary to delete 
an interval  is $O(\log n/\log\log n)$.

%We will show in Appendix A how the tree can be rebalanced after 
%updates, so that the height of $\cT$ remains $O(\log n/\log\log n)$.
{\bf Re-Balancing of the Tree Nodes.}
It remains to show how the tree can be rebalanced after 
updates, so that the height of $\cT$ remains $O(\log n/\log\log n)$.
We implement the base tree $\cT$ as a weight-balanced B-tree~\cite{AV03}
with branching parameter $\phi$ and leaf parameter $\phi$ for 
$\phi=\log^{\eps} n$. We assume that the constant  $\eps< 1/8$. 

When a segment $s$ is  deleted, we simply mark  it as deleted. 
After $n/2$ deletions, we re-build the base tree and all data structures 
stored in the nodes.   This can be done in $O(n\log n/\log \log n)$ time. 

Now we show how insertions can be handled. 
We denote by the weight of $u$ the total number of elements in the leaf descendants of $u$. The weight $n_u$ of $u$ also equals to the number of 
segment identifiers in $\Comp(u)$.  
Our  weight-balanced B-tree is organized in such way that 
$n_u=\Theta(\phi^{\ell+1})$, where 
$\phi=\log^{\eps}n$ and $\ell$ is the level of a node $u$. 
A node is split after $\Theta(\phi^{\ell+1})$ insertions; when a node 
$u$ is split into $u'$ and $u''$, the ranges of the other nodes in 
the base tree do not change; see~\cite{AV03} for a detailed description of 
node splitting. 
 We will show  that 
all relevant data structures can be re-built in $O(n_u)$ time. 
Hence, the amortized cost of splitting a node is $O(1)$. 
When a segment is inserted, it is inserted into $O(\log n/\log \log n)$ 
nodes of the base tree. 
Hence, the total amortized cost of all splitting operations caused by 
inserting a segment $s$ into our data structure is $O(\log n/\log \log n)$. 

It remains to show how secondary data structures are updated when a node 
$u$ is split. Let $w$ be the parent of $u$. The total number of segments in 
$S(u')$ and $S(u'')$ does not exceed $|S(u)|=O(n_u)$. Hence, we can construct 
data structures $D(u')$, $D(u'')$ and lists $\Comp(u')$, 
$\Comp(u'')$   with all block data structures in $O(|S(u)|)=O(n_u)$ time.
We can find the positions of all elements $e\in \Comp(u')\cup \Comp(u'')$ 
in $\Comp(w)$, update their identifiers in $\Comp(w)$, and update 
all auxiliary data structures in $O(n_u)$ time. 
\begin{figure}[tb]
\centering
\begin{tabular}{ccc}
\includegraphics[width=.4\textwidth]{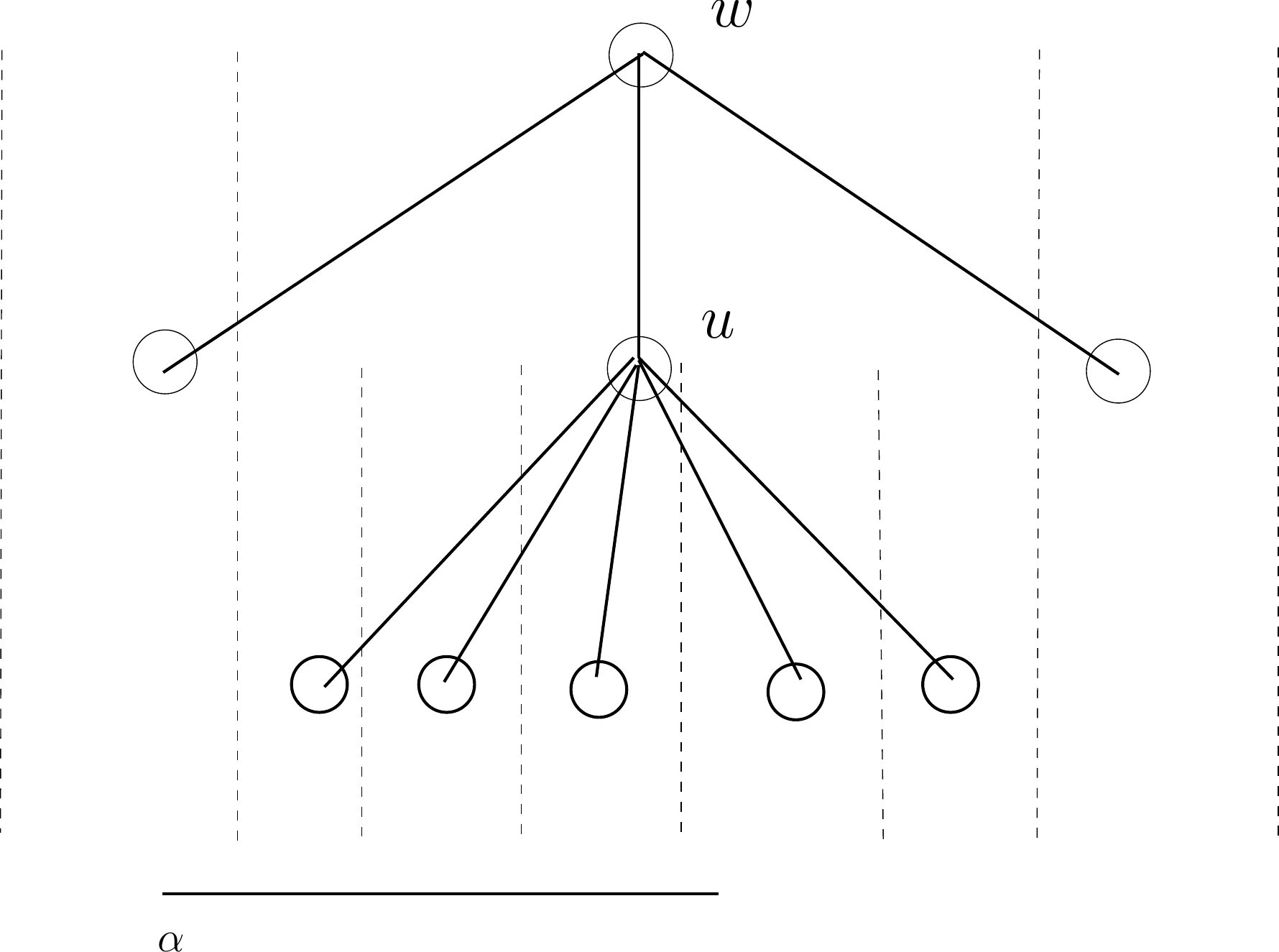} & \hspace*{2cm} &
\includegraphics[width=.4\textwidth]{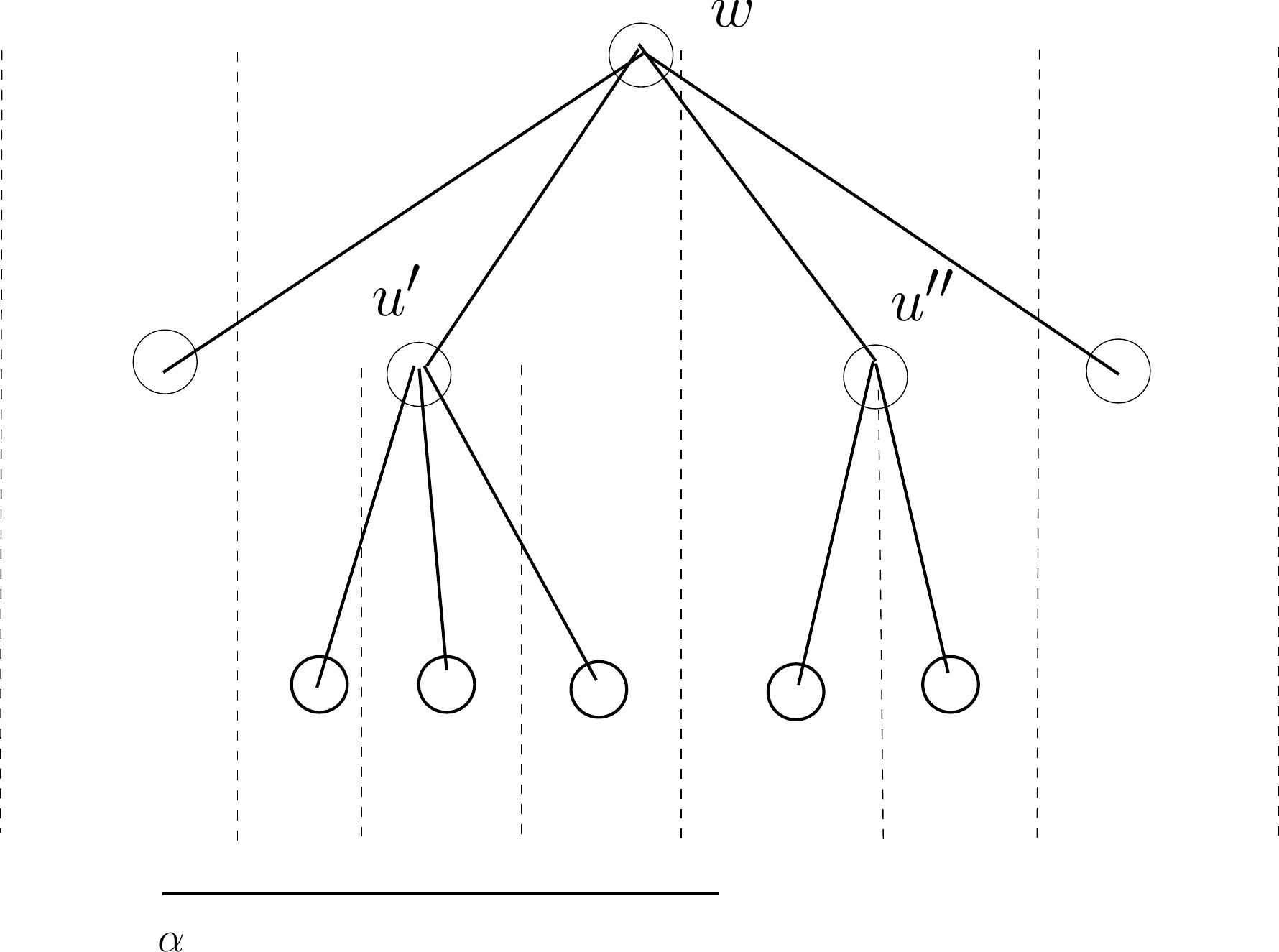} \\
{\bf (a)}   &  & {\bf(b)} \\
\end{tabular}
\caption{\label{fig:split} 
Node $u$ is split into $u'$ and $u''$. 
Segment $\alpha$ is moved from $S(u)$ to $S(w)$ after splitting of a node $u$.
}
\end{figure}

Some of the intervals stored in $S(u)$ can be moved from $S(u)$ to $S(w)$: 
if an interval $s \in S(u)$ does not cover $rng(u)$ but covers $rng(u')$ 
or $rng(u'')$, then $s$ must be stored in $S(w)$ after splitting. 
See Fig.~\ref{fig:split} for an example. 
The total number of elements in $\Comp(w)$ is 
$\Theta(\phi^{\ell+2})=\Theta(n_u\cdot \phi)$; hence,  the total number of 
blocks in $\Comp(w)$ is $o(n_u/\log^3 n)$. 
Identifiers of all  segments in  $S(u)$ are already stored in $\Comp(w)$. 
We can update all block data structures in $O(1)$ time per segment. 
Every update of the data structure $D(w)$ takes $O(\log \log n)$ time. 
But the total number of insertions into $D(w)$ does not exceed the number 
of blocks in $D(w)$. Therefore the total time needed to update $D(w)$ is
$O(n_u\log^{\eps-3}n)=o(n_u)$.

Suppose that $u$ was the $k$-th child of $w$. Some segments stored in 
$S_{lk}(w)$ for $l<k$ can be moved to $S_{l(k+1)}(w)$, and some segments in 
$S_{kr}(w)$ for $r>k$  can be moved to $S_{(k+1)r}(w)$.  Since the total 
number of blocks in $\Comp(w)$ is $O(n_u/\log^3 n)$,  all updates of 
data structures $D_{ij}(w)$ take $O((n_u/\log^3 n)\log\log n)=o(n_u)$ time. 
At most $n_u$  segments were stored in $S_{lk}(w)$ and $S_{kr}(w)$ 
before the split of $u$. Hence, we can update all block data structures 
in $O(n_u)$ time.  

Finally, we must change the identifiers of segments stored in $\Comp(w)$
and data structures $G(B)$.  Recall that an identifier indicates in which 
children of $w$ a segment $s$ is stored.  
Suppose that a segment $s$ had an identifier $r>k$ in $\Comp(w)$. 
Then its identifier will be changed to $r+1$ after the split of $u$. 
The total number of segments with incremented identifiers does not exceed 
$n_w=O(n_u\cdot\phi)$. 
However, each identifier is stored in $O(\log\log n)$ bits. We can use 
this fact, and increment the values of $\sqrt{\log n}$ identifiers 
in $O(1)$ time using the following look-up table $T$. 
There are $O((\log n)^{\sqrt{\log n}})$ different sequences of $\sqrt{\log n}$ 
identifiers. For every such sequence $\alpha$ and  any $j=O(\log^{\eps}n)$,
 $T[\alpha][j]=\alpha'$. Here $\alpha'$ is the sequence that we obtain 
if every $j'>j$  in the sequence $\alpha$ is replaced with $j'+1$. 
Thus the list $\Comp(w)$ can be updated in $O(n_w/\sqrt{\log n})=O(n_u)$ 
time. 
Using the same approach, we can also update data structures $F(B)$ and $G(B)$ 
for each block $B$ in $O(|B|/\sqrt{\log n})$ time. 
Hence, the total time necessary to update all data structures in $w$ because 
some identifiers must be incremented is $O(n_u)$. 

Since all data structures can be updated in $O(n_u)$ time when a node $u$ is
split, the total amortized cost of an insertion is $O(\log n)$.

\section{Multi-Dimensional Stabbing-Max Queries}
\label{sec:multidim}
Our data structure can be extended to the case of $d$-dimensional
stabbing-max queries for $d>1$. In this section we prove the following 
theorem.
\begin{theorem}
  \label{theor:multidim}
  There exists a data structure that uses $O(n(\log n/\log \log
  n)^{d-1})$ space and answers $d$-dimensional stabbing-max queries in
  $O((\log n/\log\log n)^{d})$ time. Insertions and deletions are
  supported in $O((\log n/\log\log n)^{d}\log\log n)$ and $O((\log
  n/\log\log n)^{d})$ amortized time respectively.
\end{theorem}
 Let $(d,g)$-stabbing problem denote
the $(d+g)$-dimensional stabbing-max problem for the case when the
first $d$ coordinates can assume arbitrary values and the last $g$
coordinates are bounded by $\log^{\rho}n$, for
$\rho\leq\frac{1}{8(g+d)}$.  First, we show that the data structure of
Theorem~\ref{theor:1d} can be modified to support $(1,g)$-stabbing
queries for any constant $g$. Then, we will show how arbitrary
$d$-dimensional queries can be answered. Throughout this section
$proj_i(s)$ denotes the projection of a rectangle $s$ on the $i$-th
coordinate. We denote by $b_i(s)$ and $e_i(s)$ the $i$-th coordinates
of endpoints in a rectangle $s$, so that $proj_i(s)=[b_i(s),e_i(s)]$.

{\bf A Data Structure for the $(1,g)$-Stabbing Problem.}  The solution of
the $(1,g)$-stabbing problem is very similar to our solution of the
one-dimensional problem.  We construct the base tree $\cT$ on the
$x$-coordinates of all rectangles. $\cT$ is defined as in
section~\ref{sec:stabover}, but we assume that $\eps\leq \rho$.  Sets
$S(u)$, $S_{lr}(u)$, and $\oS(u)$ are defined as in
sections~\ref{sec:stabover},~\ref{sec:compact}.  We define
$S_{lr}[j_1,\ldots,j_g,\ldots,j_{2g}](u)$ as the set of rectangles $s$
in $S_{lr}(u)$ such that $b_2(s)=j_1$, $\ldots$, $b_{g+1}(s)=j_g$,
$e_2(s)=j_{g+1}$, $\ldots$, $e_{g+1}(s)=j_{2g}$.

A rectangle identifier for a rectangle $s$ stored in a list $\Comp(u)$
consists of one or two tuples with $2g+1$ components. The first
component of each tuple is an index $j$, such that $s$ belongs to the
$j$-th child of $u$; remaining $2g$ components are the last $g$
coordinates of the rectangle endpoints. The data structure $M(u)$
contains the maximum priority rectangle in
$S_{lr}[j_1,\ldots,j_{2g}](u)$ for each $l,r,j_1,\ldots, j_{2g}$.
$M(u)$ can find for any $(f,h_1,\ldots, h_g)$ the highest priority
rectangle in all $S_{lr}[j_1,\ldots,j_{2g}](u)$, such that $l\leq
f\leq r$, and $j_i\leq h_i\leq j_{g+i}$ for $1\leq i \leq g$.  We can
implement $M(u)$ as for the one-dimensional data structure.  Data
structure $D(u)$ consists of VEB structures
$D_{lr}[j_1,\ldots,j_{2g}](u)$.  If a block $B$ of $\Comp(u)$ contains
a rectangle from $S_{lr}[j_1,\ldots,j_{2g}](u)$, then $\clab(B)$ is
stored in $D_{lr}[j_1,\ldots,j_{2g}](u)$.

We can answer a query $q=(q_x,q_1,\ldots,q_g)$ by traversing the
search path $\pi=v_0,\ldots, v_R$ for $q_x$ in the base tree $\cT$. As
in Theorem~\ref{theor:1d}, we start at the leaf node $v_0$ and find
the maximal rectangle stored in $S(v_l)$ that is stabbed by  the query point
$q$.  The search procedure in nodes $v_1,\ldots, v_R$ is organized in
the same way as the search procedure in section~\ref{sec:compact}:
The rectangle $s(v_i)$ is defined as in section~\ref{sec:compact}.  If
$v_{i-1}$ is the $f$-th child of $v_i$, we answer the query
$(f,q_1,\ldots,q_g)$ using the data structure $M(v_i)$.  Then we visit
the block $B_m$ returned by $M(v_i)$ and find the last rectangle $s_m$
in $B_m$ with an identifier
$(f,b_2(s_m),\ldots,b_{g+1}(s_m),e_2(s_m),\ldots,e_{g+1}(s_m))$ such
that $b_{i+1}(s_m)\leq q_i\leq e_{i+1}(s_m)$ for $1\leq i \leq g$.
Finally we compare $s_m$ with the rectangle $s(v_{i-1})$ by comparing
their positions in $\Comp(v_i)$. Thus we can find the position of
$s(v_i)$ in $\Comp(v_i)$.  When we reach the root $v_R$, $s(v_R)$ is
the highest priority segment that is stabbed by $q$.

Updates of the list $\Comp(u)$ and all auxiliary data structures are
implemented as in section~\ref{sec:rebal}.
\begin{lemma}
  \label{lemma:1g}
  There exists a data structure that answers $(1,g)$-stabbing queries
  in $O(\log n/\log\log n)$ time and uses $O(n)$ space. Insertions and
  deletions are supported in $O(\log n)$ and $O(\log n/\log \log n)$
  amortized time respectively.
\end{lemma}

\tolerance=1500 {\bf A Data Structure for the $(d,g)$-Stabbing
  Problem.}  The result for $(1,g)$-stabbing can be extended to
$(d,g)$-stabbing queries using the following Lemma.
\begin{lemma}
  \label{lemma:dg}
  Suppose that there is a $O(n(\log n/\log \log n)^{d-2})$ space data
  structure $D_1$ that answers $(d-1,g+1)$-stabbing queries in $O((\log
  n/\log \log n)^{d-1})$ time; $D_1$ supports insertions and deletions
  in $O((\log n/\log \log n)^{d-1}\log\log n)$ and $O((\log n/\log
  \log n)^{d-1})$ amortized time
  respectively. \\
  Then there exists a $O(n(\log n/\log \log n)^{d-1})$ space data
  structure $D_2$ that answers $(d,g)$-stabbing queries in $O((\log
  n/\log \log n)^{d})$ time; $D_2$ supports insertions and deletions
  in amortized time $O((\log n/\log \log n)^{d}\log\log n)$ and 
  $O((\log n/\log \log n)^{d})$  respectively.
\end{lemma}

The main idea is to construct the base tree $\cT$ on the $d$-th
coordinates of rectangles and store a data structure for
$(d-1,g+1)$-stabbing queries in each tree node. 
The tree is organized as in section~\ref{sec:stabover}
and the first part of this section.  Let $q_d$ denote the $d$-th
coordinate of a point.  Leaves contain $d$-th coordinates of rectangle
endpoints.  A rectangle $s$ is stored in a set $S(u)$ for a leaf $u$
if $proj_d(s)\cap rng(u)\not=\emptyset$ and $rng(u)\not\subset
proj_d(s)$.  A rectangle $s$ is stored in a set $S(u)$ for an internal
node $u$ if $rng(u)\not\subset proj_d(s)$ and $rng(u_i)\subset
proj_d(s)$ for at least one child $u_i$ of $u$.  $S_{ij}(u)$ is the
set of all intervals $s\in S(u)$ such that $rng(u_f)\subset proj_d(s)$
for a child $u_f$ of $u$ if and only if $i\leq f\leq j$.  We store a
data structure $\cD(u)$ for $(d-1,g+1)$-stabbing queries in each
internal node $u$. For a rectangle $s\in S_{lr}(u)$, $\cD(u)$ contains
a rectangle $s'$ such that $proj_g(s)=proj_g(s')$ for $g\not=d$ and
$proj_d(s')=[l,r]$.  In other words, we replace the $d$-th coordinates
of $s$'s endpoints with $l$ and $r$.
%For each leaf node $u_l$, we store a data structure $\cE(u_l)$ that 
%supports stabbing-max queries on $S(u_d)$ in $O(\log^{1-\eps}n)$ time. 

A query $q$ is answered by traversing the search path for $q$. 
For a leaf node $v_0$, we examine all rectangles in $S(v_0)$ and 
find the highest priority rectangle in $S(v_0)$. 
In an internal node $v_i$, $i\geq 1$, we answer the query 
$q(v_i)=(q_1,\ldots,q_{d-1},f,q_{d+1},\ldots,q_{d+g})$ using the data structure 
$\cD(v_i)$. The query $q(v_i)$ is obtained from $q$ by replacing 
the $d$-th coordinate with an index $f$, such that $v_{i-1}$ is the $f$-th 
child of 
$v_i$. When a node $v_i$ is visited, our procedure finds the highest priority 
rectangle $s(v_i)$ that is stored in $S(v_i)$ and is stabbed by $q$. 
All rectangles that are stabbed by $q$ are stored in some $S(v_i)$. 
Hence, the highest priority rectangle stabbed by $q$ is the maximum rectangle 
among all $s(v_i)$.
Since our procedure spends  $O((\log n/\log\log n)^{d-1})$ time in each 
node, the total query time is $O((\log n/\log\log n)^{d})$. 

When a rectangle $s$ is inserted or deleted, we update $O(\log n/\log\log n)$ 
data structures $\cD(u)$ 
in which $s$ is stored. Hence, deletions and insertions  are supported in 
$O((\log n/\log\log n)^{d})$ and $((\log n/\log\log n)^{d}\log\log n)$ time
respectively. We can re-balance the base tree using a procedure that is similar 
to the procedure described in section~\ref{sec:rebal}.

{\bf Proof of Theorem~\ref{theor:multidim}.}
Theorem~\ref{theor:multidim} follows easily from Lemmas~\ref{lemma:1g}
 and~\ref{lemma:dg}.
By Lemma~\ref{lemma:1g}, there exists a linear space data structure
that answers $(1,d-1)$-stabbing queries in $O(\log n/\log\log n)$
time.  We obtain the main result for $d$-dimensional stabbing-max
queries by applying the result of Lemma~\ref{lemma:dg} to the data
structure for $(1,d-1)$-stabbing queries.

\section{Stabbing-Sum Queries}
\label{sec:stabsum}
We can also modify our data structure so that it supports 
stabbing-sum queries: count the number of intervals stabbed by a 
query point $q$. 
\begin{theorem}
\label{theor:sum1d}
There exists a linear space data structure that answers orthogonal 
stabbing-sum queries in $O(\log n/\log \log n)$ time. This data structure 
supports insertions and deletions in $O(\log n)$ and $O(\log n/\log \log n)$ 
amortized time respectively.
\end{theorem}
The only difference with the proof  of Theorem~\ref{theor:1d} 
is that we store data structures for counting 
intervals instead of $M(u)$. We maintain a  data structure $X(u)$ in 
each internal  node $u$.
 For a query index $f$, $X(u)$ reports in $O(1)$ time the total number 
of intervals in $\cup_{l\leq f\leq r} S_{lr}(u)$. 
In every leaf node $u_l$, we maintain a  data structure $Z(u_l)$ 
that supports  stabbing sum queries on $S(u_l)$ and updates in 
$O(\log\log n)$ time.

As above, let  $\pi=v_0\ldots v_R$ denote the search path for a query 
point $q$. In every node $v_i\in \pi$, $i\geq 1$, we count 
the number $n(v_i)$ of intervals in  $\cup_{l\leq f\leq r} S_{lr}(u)$
using $X(v_i)$; the 
index $f$ is 
chosen so  that $v_{i-1}$ is the $f$-th child of $v_i$.
We also compute the number $n(v_0)$ of intervals that belong to 
$S(v_0)$ and are stabbed by $q$ using $Z(v_0)$. 
An interval $s$ is stabbed by $q$ either if $s$ is stored in $S(v_i)$,
$i\geq 1$, and $rng(v_{i-1})\subset s$ or if $s$ is stored in $S(v_0)$ and 
$s$ is stabbed by $q$. Hence, $q$ stabs $\sum_{v_i\in \pi} n(v_i)$ intervals.
Each $n(v_i)$, $i\geq 1$, can be found in $O(1)$ time and $n(v_0)$ 
can be found in $O(\log \log n)$ time. Hence,  a query can be answered 
in $O(\log n/\log\log n)$ time.

Now we turn to the description of  the data structure $X(u)$.
Following the idea 
of~\cite{PD04}, we store information about the recent  updates 
in a word $B$; the array $A$ reflects the state of the structure 
before recent updates. $A$ is a static array that is re-built after 
$\log^{2\eps}n$ updates of $X(u)$.
When we construct $A$, we  set $A[f]=\sum_{l\leq f\leq r} |S_{lr}(u)|$. 
The word $B$ contains one integer value $m(l,r)$ 
for each pair $l\leq r$ ($m(l,r)$ can also be negative, but the absolute 
value of each $m(l,r)$ is bounded by $O(\log^{\eps}n)$). 
When a segment is inserted into (deleted from) $S_{lr}(u)$, we 
increment (decrement) the value of $m(l,r)$ by $1$. 
We can find $\sum_{l\leq f\leq r} m(l,r)$ in $O(1)$ time using a look-up table. 
After $\log^{\eps} n$ updates we rebuild the array $A$ and set 
all $m(l,r)=0$. The amortized cost of rebuilding $A$ is $O(1)$.

The total number of segments in $\cup_{l\leq f\leq r} S_{lr}(u)$ equals to 
$A[f]+\sum_{l\leq f\leq r} m(l,r)$. Hence, a query to $X(u)$ is answered 
 in $O(1)$ time. 
We implement $A$ in such way that each entry $A[i]$ uses $O(\log |S(u)|)$ 
bits. Thus each data structure $X(u)$ uses $O(\log^{\eps}n \log |S(u)|)$
bits and all $X(u)$, $u\in T$, use $O(n\log n)$ bits in total. 

It remains to describe the data structure $Z(u)$. 
Let $E(u)$ be the set that contains endpoints of all intervals 
in $S(u)$. Since $u$ is a leaf 
node, $S(u)$ contains $O(\log^{2\eps} n)$ elements.  Hence, it takes 
$O(\log \log n)$ time to  find the rank $r(q)$ of $q$ in $E(u)$.
For each $1\leq i\leq |E|$, $C[i]$ equals to  the number of intervals 
that are stabbed by  a point $q$ with rank $r(q)=i$.  
The array $C$ enables us to count intervals stabbed by $q$ in $O(1)$ 
time if the rank of $q$ in $E(u)$ is known. 
Since $|E|=O(\log^{\eps} n)$ and $C[i]=O(\log^{\eps} n)$, the array 
$C$ uses $O(\log^{\eps}n\log\log n)$ bits of  memory. 
When a set $S(u)$ is updated, 
we can update $C$ in $O(1)$ time using a look-up table.
Thus $Z(u)$ uses linear space and supports both queries and updates 
in $O(\log\log n)$ time.

\newpage

\section*{Appendix A. Auxiliary Data Structures}
%\label{sec:details}
Using bitwise operations and table look-ups, we can implement the 
 data structure $M(u)$ and the block data structures so that queries 
are supported in constant time.
 
{\bf Data Structure $M(u)$.}
Let $\cM(u)$ be the set of all interval priorities stored in $M(u)$. 
Recall that $\cM(u)$ contains an  element $\max_{ij}$ for each set 
$S_{ij}(u)$, where 
$\max_{ij}$ is the highest priority interval (or its block label) stored 
in $S_{ij}(u)$. 
For simplicity, we assume that $\max_{ij}=-\infty$ if 
$S_{ij}=\emptyset$. 
Let the rank of $x$ in $\cM(u)$ be defined as 
$\rank(x,\cM(u))=|\{\,e\in \cM(u)\,|\,e\leq x\,\}|$. 
Let $\cM'(u)$ be the set in which every element of $\cM(u)$ is replaced 
with its rank in $\cM(u)$.  
That is, for each $\max_{ij}$ stored in $\cM(u)$ we store 
$\max'_{ij}=\rank(\max_{ij},\cM(u))$ in the set $\cM'(u)$. 

Each of $O(\log^{1/4}n)$ elements in $\cM'(u)$ 
is bounded by $O(\log^{1/4}n)$. 
Hence, we can store $\cM'(u)$ in one bit sequence $W$; $W$ consists 
of $O(\log^{1/4}n\log\log n)$ bits and  fits into a machine word.  We can 
store a table $\Tbl_1$ with an entry for each possible value of $W$ and
 for each $f$. 
The entry $\Tbl[W][f]$ contains the pair $<\!l,r\!>$, such that 
$\max'_{lr}$ is the highest value in the set 
$\{\,\max'_{ij}\,|\, i\leq f\leq j\,\}$.

%Suppose that a set $\cM(u)$ is stored in our data structure $M(u)$.
Obviously, $\max_{ab}<\max_{cd}$ if and only if $\max'_{ab}<\max'_{cd}$. 
Therefore a query $f$ can be answered by looking up the value 
$<\!l,r\!>=\Tbl[W][f]$ for $W=\cM'(u)$ and returning the element 
$\max_{lr}$.

When $M(u)$ is updated, the value of some $\max_{ij}$ is changed. 
 We store all elements of $\cM(u)$ in an atomic heap~\cite{FW94,Th99} 
that supports predecessor queries and updates in $O(1)$ time.  
Hence, the new rank of $\max_{ij}$ 
in $\cM(u)$ can be found in $O(1)$ time.  
If the rank of $\max_{ij}$ has changed, the ranks of all other elements 
in $\cM(u)$ also change.  
Fortunately, $\max'_{ij}$ can assume only $O(\log^{1/4}n)$ values. 
There are  $O(\log^{1/4}n)$ elements $\max'_{ij}$ in a set $\cM'(u)$,  
and there are $2^{O(\log^{1/4}n)}=o(n)$ different sets $\cM'(u)$.
Besides that, each $W=\cM'(u)$ fits into one word. 
Hence, we can update the set $\cM'(u)$ using a look-up in a table 
$\Tbl_1$. 
Suppose that the rank of the $f$-th element in a set $\cM'(u)$ is 
changed to $r$. Then the new set $\cM'(u)$ is stored  
in an entry $\Tbl_1[W][f][r]$ of $\Tbl_1$; here $W$ is the bit sequence 
that corresponds to the old set $\cM'(u)$. 
Thus an update of $M(u)$ can be implemented in $O(1)$ time. 

We need only one instance of tables $\Tbl$ and $\Tbl_1$ for all nodes $u$ 
of the base tree. Both tables use $o(n)$ space and can be initialized 
in $o(n)$ time.

{\bf Data Structure $F(B)$.}
%We give a sketch of the data structure $F(B)$ in Appendix C. 
The data structure $F(B)$ for a block $B\in \Comp(u)$ is implemented as a 
B-tree $\bT_F$. Every leaf of $\bT_F$ contains $\Theta(\rho)$ identifiers and each internal node has degree $\Theta(\rho)$ for 
$\rho=\sqrt{\log n}$. Recall that each identifier consists of at most two 
indices $i_1$, $i_2$  such that the corresponding interval is stored in 
the $i_1$-th and $i_2$-th children of the node $u$ in the base tree. 
 All identifiers stored in a leaf node can be packed 
into $O(\rho\log\log n)$ bits. In every internal node $\nu$ of $\bT_F$, we store the bit sequence  $W(\nu)$.  For every child $\nu_i$ and for every possible value 
of $j$, $W(\nu)$ contains  the total 
number of indices $j$  in the $i$-th child $\nu_i$ of $\nu$; 
$W(\nu)$ consists of  $O(\rho\log \log n)$ bits. 
%Using a look-up table similar to tables used in the data structure $M(u)$, we 
%can answer the following queries for every possible $W(\nu)$ in constant 
%time.  
Using a look-up table, we can count for any $W(\nu)$ and for any $i$ 
the total number of elements stored in 
the children $\nu_1,\ldots,\nu_{i-1}$ of $\nu$. 
For any $W(\nu)$ and $i$, we can also count the total number of 
indices $j$ in 
the children $\nu_1,\ldots,\nu_{i-1}$ of $\nu$. 
For any $k\leq 2\log^3 n$ and any $W(\nu)$, $j$,   we can find the largest 
$i$ such that the total number of indices $j$ in 
$\nu_1\ldots,\nu_i$ does not exceed $k$. 
Look-up tables that support such queries use $o(n)$ space and can be 
initialized in $o(n)$ time. 
We need only one instance of each table for all blocks $B$ and all $F(B)$.

All queries and updates of  a data structure $F(B)$ can be 
processed by traversing a path 
in $\bT_F$. Using the above described look-up-tables, we spend 
only constant time in each node of $\bT_F$; hence, queries and updates 
are supported in $O(1)$ time.
Details will be given in the full version of this paper. 
Data structures $G(B)$ and $H(B)$ can be implemented in a similar way.

%We can find the identifier of the $k$-th element in $\bT_F$ by traversing 
%a path in $\bT_F$. We can also answer $\rank$ and $\select$ queries in the 
%same way. Details will be given in the full version of this paper. 

\end{document}